\documentclass[Journal,twocolumn,10pt]{IEEEtran}
\usepackage{amssymb}
\usepackage{amsmath}

\usepackage{algorithm}
\usepackage{algpseudocode}
\usepackage{algorithmicx}

\newcommand{\comment}[1]{}

\newtheorem{theorem}{Theorem}[section]
     \newtheorem{lemma}[theorem]{Lemma}

     \newcommand{\qed}{\nobreak \ifvmode \relax \else
           \ifdim\lastskip<1.5em \hskip-\lastskip
           \hskip1.5em plus0em minus0.5em \fi \nobreak
           \vrule height0.75em width0.5em depth0.25em\fi}
\begin{document}

% paper title

\title{Write-Once-Memory Codes by Source Polarization }

\author{\authorblockN{Xudong Ma \\}
\authorblockA{Pattern Technology Lab LLC, Delaware, U.S.A.\\
Email: xma@ieee.org} }

\maketitle

\begin{abstract}

We propose a new Write-Once-Memory (WOM) coding scheme based on source polarization. By applying a source polarization transformation on the to-be-determined codeword, the proposed WOM coding scheme encodes information into the bits in the high-entropy set. We prove in this paper that the proposed WOM codes are capacity-achieving. WOM codes have found many applications in modern data storage systems, such as flash memories.

\end{abstract}

\section{Introduction}

\label{section_introduction}

WOM codes have attracted much attention in the recent years, due to the fact that they have found applications in flash memories. Flash memories usually can endure limited numbers of erasure operations. For modern TLC-type flash memories, the average acceptable numbers of erasure operations can be as small as 5000. Therefore, it is desirable that each piece of flash memory may be reused to record another piece of data without using erasure operations.

If no erasure operation is used, then the flash memories become the so called ``Write-Once-Memory'' (WOM). Examples of WOM include optical disks, paper taps, and punch cards. Each paper tape may be punched at multiple positions. Each punched position may represent a bit value 1 and each unpunched position may represent a bit value 0. Each bit of the paper tape may be turned from 0 to 1 but never from 1 to 0. Similarly, for SLC-type flash memories, if no erasure operation is used, then each memory cell can be turned from 1 to 0, but never from 0 to 1.

The definition of WOM codes is thus as follows. Let $y_1,y_2,\ldots,y_N$ denote all the bits contained in one piece of flash memories. Each bit $y_n$ may take values 0 or 1. The WOM encoder takes the current memory states $y_1,y_2,\ldots,y_N$ and the to-be-recorded information bits $v_1,v_2,\ldots,v_M$ as inputs and outputs a codeword $x_1,x_2,\ldots,x_N$, such that $x_n=0$ if $y_n=0$ for $n=1,2,\ldots,N$. The WOM decoder takes the codeword $x_1,x_2,\ldots,x_N$ as inputs and outputs the recorded information bits $v_1,v_2,\ldots,v_M$.

WOM codes were first discussed in 1980s \cite{rivest82}. However, constructing capacity-achieving WOM codes had been an open problem for more than three decades. Surprisingly, two constructions of capacity-achieving WOM codes were proposed recently \cite{burshtein12} \cite{shipilka12}. The coding scheme in \cite{shipilka12} is based on algebraic approaches. The coding scheme in \cite{burshtein12} is based on polar codes (channel polarization).

The current paper proposes a third approach of constructing capacity-achieving WOM codes by using source polarization. The source polarization was invented by Arikan for lossless source coding \cite{arikan10}. We show in this paper that source polarization may be used for constructing WOM codes as well. We prove that the thus constructed WOM codes are capacity-achieving. The WOM codes proposed in this paper have the advantages of low-complexity software and hardware implementations.

The paper is organized as follows. In Section \ref{section_source_polar}, we will discuss the considered probability models for the WOM encoding problem. We will  prove two lemmas, which are the consequences of Theorem 1 in \cite{arikan10}. In Section \ref{section_basic_scheme}, we will present our invented WOM coding scheme. In Section \ref{section_total_var}, we will show the invented coding scheme is capacity-achieving. Finally, some concluding remarks will be presented in Section~\ref{sec_conclusion}.

We use the following notation throughout this paper. We  use
$X_{n}^{N}$ to denote a sequence of symbols $X_n, X_{n+1}, \ldots, X_{N-1}, X_{N}$. We use upper-case letters to denote random variables and lower-case letters to denote the corresponding realizations. For example, $X_i$ is one random variable, and $x_i$ is one realization of the random variable $X_i$. For a binary random variable $X$ with probabilities ${\mathbb P}(X=0)=p$, ${\mathbb P}(X=1)=1-p$, we define the entropy function $H(p)$ as follows,
\begin{align}
H(p) = p \log \left(\frac{1}{p}\right) + (1-p) \log\left(\frac{1}{1-p}\right)
\end{align}

\section{Probability Models and Auxiliary Lemmas}
\label{section_source_polar}

In this section, we will consider the following probability model of random variables $Y_{1}^N$, $X_{1}^{N}$ and $U_{1}^{N}$.  The physical meanings of $Y_{1}^{N}$, $X_{1}^{N}$ are as follows. The random variables $Y_{1}^{N}$ are the states of the memory cells. The random variables $X_{1}^{N}$ are the codeword bits. In this paper, we only consider binary cases, i.e., $Y_{1}^{N}$, $X_{1}^{N}$ are all binary random variables. However, it should be clear that the coding scheme and analysis can be easily generalized to the non-binary cases.

We assume that $Y_n$ are independent and identically distributed  (iid) with the following probability distribution $P(\cdot)$,
\begin{align}
{P}(Y_n) = \left\{
\begin{array}{ll} s, & \mbox{if } Y_n = 0 \\
1-s, & \mbox{if }Y_n = 1
\end{array}
\right.
\end{align}
where, $s$ is a real number, $0<s<1$.
Let each $X_i$ be conditionally independent of the other $X_j$ given $Y_i$.
\begin{align}
{P}(X_n|Y_n) = \left\{
\begin{array}{ll}
1,  & \mbox{if } X_n=0, Y_n=0 \\
0,  & \mbox{if } X_n=1, Y_n=0 \\
t,  & \mbox{if } X_n=0, Y_n=1 \\
1-t,  & \mbox{if } X_n=1, Y_n=1
\end{array}
\right.
\end{align}
where, $t$ is a real number, $0<t<1$.  
Let us define a matrix $G_N$ in the same way as in \cite{arikan10}.
\begin{align}
G_N = \left[\begin{array}{cc}
1 & 0 \\
1 & 1
\end{array}\right]^{\otimes n} B_N
\end{align}
where, $\otimes n$ denotes the nth Kronecker power and $B_N$ is the bit-reversal permutation. According to the definition, $G_N$ is an $N$ by $N$ square matrix. Define binary random variables $U_1, U_2, \ldots,U_N$ as follows.
\begin{align}
U_{1}^{N} = X_{1}^{N} G_N
\end{align}
where, $X_{1}^{N}$ denotes the row vector $[X_1,X_2,\ldots,X_N]$ and $U_{1}^{N}$ denotes the row vector $[U_1,U_2,\ldots,U_N]$. Clearly, $U_{1}^{N}$ is uniquely determined by $X_{1}^{N}$.

The configuration of the random variables $Y_{1}^{N}$, $X_{1}^{N}$, and $U_{1}^{N}$ is a special case of the scenarios in \cite{arikan10}, where $Y_{1}^{N}$ is the side information, $X_{1}^{N}$ is the to be compressed bit string. The following lemmas are corollaries of Theorem 1 in \cite{arikan10}.
\begin{lemma}
\label{polar_entropy_lemma}
There exists a sequence of sets $F_N$, where, $N$ are integers and go to infinity, and each $F_N$ is a subset of $\{1,2,\ldots,N\}$, such that the following hold.
\begin{itemize}
\item The cardinality $|F_N|$ of the set $F_N$ satisfies \begin{align}
\frac{|F_N|}{N} \geq (1-\epsilon_N) H(X_n|Y_n) = (1-\epsilon_N) (1-s)H(t)
\end{align}
\item For each $i\in F_N$, $H(U_i|Y_{1}^{N},U_{1}^{i-1}) \geq 1- \delta_N$
\item $\epsilon_N>0$, $\epsilon_N\rightarrow 0$ as $N\rightarrow \infty$, and
\item $\delta_N>0$, $\delta_N\rightarrow 0$ as $N\rightarrow \infty$.
\end{itemize}
\end{lemma}

\begin{lemma}
\label{polar_prob_lemma}
There exists a sequence of sets $F_N$, where, $N$ are integers and go to infinity, and each $F_N$ is a subset of $\{1,2,\ldots,N\}$, such that the following hold.
\begin{itemize}
\item The cardinality $|F_N|$ of the set $F_N$ satisfies \begin{align}
\frac{|F_N|}{N} \geq  (1-\epsilon_N) (1-s)H(t)
\end{align}
\item For each $i\in F_N$,
\begin{align}
\sum_{y_{1}^{N},u_{1}^{i-1}} P(y_{1}^{N},u_{1}^{i-1}) \left|P(U_i=0|y_{1}^{N},u_{1}^{i-1}) - \frac{1}{2}\right| \leq  \delta_N
\end{align}
\item $\epsilon_N>0$, $\epsilon_N\rightarrow 0$ as $N\rightarrow \infty$, and
\item $\delta_N>0$, $\delta_N\rightarrow 0$ as $N\rightarrow \infty$.
\end{itemize}
\end{lemma}
\begin{proof}
Prove by contradiction. Suppose not. Let $F_N$ denote a sequence of sets, such that, the cardinality $|F_N|$ of the set $F_N$ satisfies \begin{align}
\frac{|F_N|}{N} \geq (1-\epsilon_N) H(X_i|Y_i) = (1-\epsilon_N) (1-s)H(t)
\end{align}
where, $\epsilon_N>0$, $\epsilon_N\rightarrow 0$ as $N\rightarrow \infty$.
Suppose that for each such sequence $F_N$ and for each $N$, there exist at least one $i$, such that
\begin{align}
\sum_{y_{1}^{N},u_{1}^{i-1}} P(y_{1}^{N},u_{1}^{i-1}) \left|P(U_i=0|y_{1}^{N},u_{1}^{i-1}) - \frac{1}{2} \right| = a_N,
\end{align}
and $a_N>0$, $a_N$ is bounded away from $0$. That is,  there exists an $a>0$, such that $a_N\geq a$ for all $N$.

Let us define one binary variable $b(y_{1}^{N},u_{1}^{i-1})$  for each configuration $y_{1}^{N},u_{1}^{i-1}$  , such that
\begin{align}
b(y_{1}^{N},u_{1}^{i-1}) = \left\{\begin{array}{ll}
0, & \mbox{if } P(U_i=0|y_{1}^{N},u_{1}^{i-1}) < 1/2 \\
1, & \mbox{otherwise}
\end{array}
\right.
\end{align}
It can be easily checked that
\begin{align}
& \sum_{y_{1}^{N},u_{1}^{i-1}} P(y_{1}^{N},u_{1}^{i-1}) \left|P(U_i=0|y_{1}^{N},u_{1}^{i-1}) - \frac{1}{2} \right|  \\
& = \sum_{y_{1}^{N},u_{1}^{i-1}} P(y_{1}^{N},u_{1}^{i-1}) \\
& \hspace{0.5in} \times \left(\frac{1}{2} - P(U_i=b(y_{1}^{N},u_{1}^{i-1})|y_{1}^{N},u_{1}^{i-1})\right) \\
& = a_N,
\end{align}
Therefore,
\begin{align}
 & \sum_{y_{1}^{N},u_{1}^{i-1}} P\left(y_{1}^{N},u_{1}^{i-1}\right)  P\left(U_i=b(y_{1}^{N},u_{1}^{i-1})|y_{1}^{N},u_{1}^{i-1}\right)  \\
 & = \frac{1}{2}-a_N
\end{align}

Then, we have the bound in Eq. \ref{aux_lemma_bound_one},  where (a) follows from the fact that the entropy function $H(t)=H(1-t)$, (b) follows from the Jensen's inequality because the entropy function $H(\cdot)$ is concave, and (c) follows from the fact that the entropy function $H(\cdot)$ is increasing in the interval $(0,1/2)$. From the above inequality, we have
$H(U_i|Y_{1}^{N},U_{1}^{i-1})$ is bounded away from $1$. This statement contradicts Lemma \ref{polar_entropy_lemma}. Hence, the supposition is false and the Lemma is true.
\end{proof}

\begin{figure*}
\begin{align}
\label{aux_lemma_bound_one}
& H(U_i|Y_{1}^{N},U_{1}^{i-1})  = \sum_{y_{1}^{N},u_{1}^{i-1}} P(y_{1}^{N},u_{1}^{i-1}) H\left(P(U_i=0|y_{1}^{N},u_{1}^{i-1})\right) \notag \\
& \stackrel{(a)}{=} \sum_{y_{1}^{N},u_{1}^{i-1}} P(y_{1}^{N},u_{1}^{i-1}) H\left(P(U_i=b(y_{1}^{N},u_{1}^{i-1})|y_{1}^{N},u_{1}^{i-1})\right) \notag \\
& \stackrel{(b)}{\leq}  H\left(\sum_{y_{1}^{N},u_{1}^{i-1}} P(y_{1}^{N},u_{1}^{i-1}) P(U_i=b(y_{1}^{N},u_{1}^{i-1})|y_{1}^{N},u_{1}^{i-1})\right)  = H(1/2-a_N)  \stackrel{(c)}{\leq} H(1/2-a)
\end{align}
\end{figure*}

\section{The Encoding and Decoding Algorithms}
\label{section_basic_scheme}

In this section, we will first present the  encoding algorithm of the invented WOM codes. The encoding algorithm is a randomized algorithm. From Lemma \ref{polar_prob_lemma}, we can see that there exists a set $F_N$, such that for each $i\in F_N$, $P(U_i|y_{1}^{N}, u_{1}^{i-1})$ is almost an uniform distribution. We call the set $F_N$ as the high-entropy set. Our invented coding algorithm uses the bits $U_i$ with $i\in F_N$ to record the to-be-recorded information. Let the cardinality $|F_N|$ of the set $F_N$ be $M$, then the number of the to-be-record bits is $M$. Let us denote the to-be-recorded bits by $v_1,v_2,\ldots,v_M$. The encoding algorithm is shown in Algorithm \ref{polar_source_algorithm}.

\begin{algorithm}
\caption{WOM encoding by Source Polarization}\label{polar_source_algorithm}
\begin{algorithmic}[1]
\State The algorithm takes inputs
\begin{itemize}
\item the current memory states of the memory cells $y_1^N$
\item the high-entropy set $F_N$
\item the to-be-recorded message $v_1^M$
\item parameter $s$
\item parameter $t$
\end{itemize}
\State{$n\gets 1, m\gets 1$}
\Repeat
  \If{ $n\in F_N$}
    \State $u_n \gets v_m$
    \State $n\gets n+1$
    \State $m\gets m+1$
  \Else
    \State Calculate $P(U_n|y_1^N,u_{1}^{n-1})$
    \State Randomly set $u_n$ according to the probability distribution $P(U_n|y_1^N,u_{1}^{n-1})$. That is
    \begin{align}
    u_n = \left\{
    \begin{array}{ll}
    0, &  \mbox{with probability } P(U_n=0|y_1^N,u_{1}^{n-1})\\ \notag
    1, &  \mbox{with probability } P(U_n=1|y_1^N,u_{1}^{n-1}) \notag
    \end{array}
    \right.
    \end{align}

    \State $n\gets n+1$

  \EndIf

\Until{$n>N$}
\State $x_{1}^{N} \gets u_{1}^{N} \left(G_N\right)^{-1}$
\State The algorithm outputs $x_{1}^{N}$ as the WOM codeword
\end{algorithmic}
\end{algorithm}

The encoding algorithm takes the inputs including the current memory states of the memory cells $y_1^N$, the high-entropy set $F_N$, the to-be-recorded message $v_1^M$, the parameter $s$ and $t$. The algorithm then determines $u_n$ one by one from $n=1$ to $n=N$. If $n\in F_N$, then $u_n$ is set to one of the to-be-recorded information bits $v_m$. If $n\notin F_N$, then $u_n$ is randomly set to 1 with probability $P(U_n=1|y_1^N,x_{1}^{n-1})$, and $u_n$ is randomly set to 0 with probability $P(U_n=0|y_1^N,x_{1}^{n-1})$. After all the $u_n$ for $n=1, \ldots, N$ have been determined, a vector $x_{1}^{N}$ is calculated as $x_{1}^{N} = u_{1}^{N} \left(G_N\right)^{-1}$. In other words, $u_{1}^{N} = x_{1}^{N} G_N$. The algorithm finally outputs $x_{1}^{N}$ as the codeword.

The recorded information can be recovered by a rather simple and low-complexity decoding algorithm as shown in Algorithm \ref{polar_source_algorithm_decoding}. Note that the auxiliary variables $u_{1}^{N} = x_{1}^{N} G_N$. And the recorded information bits $v_1,v_2,\ldots,v_M$ are the bits of $u_{1}^{N}$ at the positions in $F_N$.

\begin{algorithm}
\caption{WOM decoding by Source Polarization}\label{polar_source_algorithm_decoding}
\begin{algorithmic}[1]
\State The algorithm takes inputs
\begin{itemize}
\item the current codeword $x_1^N$
\item the high-entropy set $F_N$
\end{itemize}
\State{$u_{1}^{N}\gets x_{1}^{N}G_N$}
\State{$n\gets 1, m\gets 1$}
\Repeat
  \If{ $n\in F_N$}
    \State $v_m \gets u_n$
    \State $n\gets n+1$
    \State $m\gets m+1$
  \Else
    \State $n\gets n+1$
  \EndIf
\Until{$n>N$}
\State The algorithm outputs $v_{1}^{M}$ as the decoded information
\end{algorithmic}
\end{algorithm}

\section{Capacity Achieving Proof}
\label{section_total_var}

In this section, we will prove that the proposed WOM codes are capacity achieving. We have defined an artificial probability distribution $P$ in Section \ref{section_source_polar}.
\begin{align}
{P}(Y_i) = \left\{
\begin{array}{ll} s, & \mbox{if } Y_i = 0 \\
1-s, & \mbox{if }Y_i = 1
\end{array}
\right.
\end{align}
\begin{align}
{P}(X_i|Y_i) = \left\{
\begin{array}{ll}
1,  & \mbox{if } X_i=0, Y_i=0 \\
0,  & \mbox{if } X_i=1, Y_i=0 \\
t,  & \mbox{if } X_i=0, Y_i=1 \\
1-t,  & \mbox{if } X_i=1, Y_i=1
\end{array}
\right.
\end{align}
\begin{align}
{P}(U_{1}^{N}|Y_{1}^{N},X_{1}^{N}) = \left\{
\begin{array}{ll}
1,  & \mbox{if } U_{1}^{N}= X_{1}^{N}G_N  \\
0,  & \mbox{otherwise}
\end{array}
\right.
\end{align}

On the other hand, the randomized encoding algorithm in Section \ref{section_basic_scheme} may be considered as a random process. The random process induces a probability distribution $Q$.
\begin{align}
{Q}(Y_n) = \left\{
\begin{array}{ll} s, & \mbox{if } Y_n = 0 \\
1-s, & \mbox{if }Y_n = 1
\end{array}
\right.
\end{align}
\begin{align}
{Q}(U_i|Y_{1}^{N},U_{1}^{i-1}) = \left\{
\begin{array}{ll}
1/2,  & \mbox{if } i\in F_N \\
{P}(U_i|Y_{1}^{N},U_{1}^{i-1})  ,  & \mbox{otherwise }
\end{array}
\right.
\end{align}
\begin{align}
{Q}(X_{1}^{N}|Y_{1}^{N},U_{1}^{N}) = \left\{
\begin{array}{ll}
1,  & \mbox{if } U_{1}^{N}= X_{1}^{N}G_N  \\
0,  & \mbox{otherwise}
\end{array}
\right.
\end{align}
For example, $Q(U_n=1)$ is the probability that $u_n$ is set to $1$ by using the randomized encoding algorithm. Note that we have assumed that the to-be-record bits $v_m$ are random and equally probable.

We need the following telescoping expansion in our latter discussions, 
\begin{lemma}
\label{tele_lemma}
\begin{align}
& \prod_{n=1}^{N}A_n - \prod_{n=1}^{N} B_n \\ 
& = \sum_{i=1}^{N}  \left(\left(\prod_{n=1}^{i-1} A_n\right) (A_i-B_i)
\left(\prod_{n=i+1}^{N} B_n\right)\right)
\end{align}
\end{lemma}

One important step of the capacity-achieving proof is the following bound on the total variation distance between the probability distributions $P$ and $Q$.
\begin{lemma}
\label{total_var_lemma}
The total variation distance between $P$ and $Q$ can be bounded as follows.
\begin{align}
& \sum_{y_{1}^{N}, u_{1}^{N}} \left|P(y_{1}^{N}, u_{1}^{N})-Q(y_{1}^{N}, u_{1}^{N})\right| \leq 2N\delta_N
\end{align}
where, $\delta_N \rightarrow 0$, as $N \rightarrow \infty$.
\end{lemma}
\begin{proof}
We have the bound in Eq. \ref{pq_bound}.
where
\begin{itemize}
\item (a) follows from the lemma \ref{tele_lemma};
\item (b) follows from the definition of the probability distribution $Q$;
\item (c) follows from the fact the absolute value of a sum is always less than or equal to the sum of absolution values;
\item (d) is resulting from a change of the order of summation;
\item (e) is resulting from a change of the order of summation;
\item (f) follows from the fact that the summation of probability $Q$ is $1$;
\item (g) follows from the multiplication rule of probabilities and conditional probabilities; and
\item (h) follows from lemma \ref{polar_prob_lemma}
\end{itemize}
The lemma is thus proven.
\end{proof}

\begin{figure*}
\begin{align}
\label{pq_bound}
& \sum_{y_{1}^{N}, u_{1}^{N}} \left|P(y_{1}^{N}, u_{1}^{N})-Q(y_{1}^{N}, u_{1}^{N})\right| \notag \\ \notag
& = \sum_{y_{1}^{N}, u_{1}^{N}} P(y_{1}^{N}) \left| \prod_{n=1}^{N}  P(u_n | y_{1}^{N}, u_{1}^{n-1})-
\prod_{n=1}^{N} Q(u_n|y_{1}^{N}, u_{1}^{n-1})\right| \\ \notag
& \stackrel{(a)}{=} \sum_{y_{1}^{N}, u_{1}^{N}} P(y_{1}^{N}) \left| \sum_{i=1}^{N} \left[
\left(\prod_{n=1}^{i-1}  P(u_n | y_{1}^{N}, u_{1}^{n-1})\right)  \left(
P(u_i | y_{1}^{N}, u_{1}^{i-1}) - Q(u_i | y_{1}^{N}, u_{1}^{i-1})
\right)
\left(\prod_{n=i+1}^{N} Q(u_n|y_{1}^{N}, u_{1}^{n-1})\right)\right]\right| \\ \notag
& \stackrel{(b)}{=} \sum_{y_{1}^{N}, u_{1}^{N}} P(y_{1}^{N}) \left| \sum_{i\in F} \left[
\left(\prod_{n=1}^{i-1}  P(u_n | y_{1}^{N}, u_{1}^{n-1})\right)  \left(
P(u_i | y_{1}^{N}, u_{1}^{i-1}) - 1/2
\right)
\left(\prod_{n=i+1}^{N} Q(u_n|y_{1}^{N}, u_{1}^{n-1})\right)\right]\right| \\ \notag
& \stackrel{(c)}{\leq} \sum_{y_{1}^{N}, u_{1}^{N}} P(y_{1}^{N}) \sum_{i\in F} \left|
\prod_{n=1}^{i-1}  P(u_n | y_{1}^{N}, u_{1}^{n-1})  \left(
P(u_i | y_{1}^{N}, u_{1}^{i-1}) - 1/2
\right)
\prod_{n=i+1}^{N} Q(u_n|y_{1}^{N}, u_{1}^{n-1})\right| \\ \notag
& \stackrel{(d)}{=} \sum_{y_{1}^{N}, u_{1}^{N}} \sum_{i\in F}
 P(y_{1}^{N})  \left|
\left(\prod_{n=1}^{i-1}  P(u_n | y_{1}^{N}, u_{1}^{n-1})\right)  \left(
P(u_i | y_{1}^{N}, u_{1}^{i-1}) - 1/2
\right)
\left(\prod_{n=i+1}^{N} Q(u_n|y_{1}^{N}, u_{1}^{n-1})\right)\right| \\ \notag
& \stackrel{(e)}{=} \sum_{i\in F} \sum_{y_{1}^{N}, u_{1}^{i}} \sum_{u_{i+1}^{N}}
 P(y_{1}^{N})  \left|
\left(\prod_{n=1}^{i-1}  P(u_n | y_{1}^{N}, u_{1}^{n-1})\right)  \left(
P(u_i | y_{1}^{N}, u_{1}^{i-1}) - 1/2
\right)
\left(\prod_{n=i+1}^{N} Q(u_n|y_{1}^{N}, u_{1}^{n-1})\right)\right| \\ \notag
& \stackrel{(f)}{=} \sum_{i\in F} \sum_{y_{1}^{N}, u_{1}^{i}}
 P(y_{1}^{N})  \left|
\left(\prod_{n=1}^{i-1}  P(u_n | y_{1}^{N}, u_{1}^{n-1})\right)  \left(
P(u_i | y_{1}^{N}, u_{1}^{i-1}) - 1/2
\right)\right| \\ \notag
& \stackrel{(g)}{\leq} \sum_{i\in F} \sum_{y_{1}^{N}, u_{1}^{i}}
P(y_{1}^{N}, u_{1}^{i-1})  \left|
P(u_i | y_{1}^{N}, u_{1}^{i-1}) - 1/2
\right| \\ 
& \stackrel{(h)}{\leq } 2N \delta_N
\end{align}
\end{figure*}

\begin{theorem}
For the random encoding algorithm \ref{polar_source_algorithm}, assume the number of memory cells is $N$. Then, the expected ratio of the number of memory cells turning from state 1 to state 0 to $N$ goes to $(1-s)t$, as $N$ goes to infinity. The expected ratio of recorded information bits to $N$ goes to $(1-s)H(t)$, as $N$ goes to infinity. Thus, the proposed WOM codes are capacity-achieving.
\end{theorem}
\begin{proof}
Let us define the following indicator function $I(x_n,y_n)$.
\begin{align}
I(x_n,y_n) = \left\{\begin{array}{ll}
1, & \mbox{if }x_n = 0, y_n = 1 \\
0, & \mbox{otherwise}
\end{array}
\right.
\end{align}
The expected number of memory cells turning from state $1$ to state $0$ is thus
\begin{align}
&{\mathbb E}_{Q}\left[\sum_{n=1}^{N} I(x_n,y_n)\right] \notag \\
& =\sum_{y_{1}^{N}, u_{1}^{N},x_{1}^{N}} Q(y_{1}^{N}, u_{1}^{N},x_{1}^{N}) \sum_{n=1}^{N} I(x_n,y_n)
\end{align}

On the other hand, we have the bound in Eq. \ref{pq_exp_bound}, where (a) follows from that fact that the absolution value of a sum is always less than or equal to the sum of absolute values; (b) follows from the fact that the indicator function $I(x_n,y_n)\leq 1$; (c) follows from the fact that $x_{1}^{N}$ is deterministic given $y_{1}^{N}$ and $u_{1}^{N}$; and (d) follows from Lemma \ref{total_var_lemma}. As a consequence, as $N$ goes to infinity,
\begin{align}
\frac{1}{N}{\mathbb E}_{Q}\left[\sum_{n=1}^{N} I(x_n,y_n)\right] \rightarrow
\frac{1}{N}{\mathbb E}_{P}\left[\sum_{n=1}^{N} I(x_n,y_n)\right]
\end{align}

On the other hand,
\begin{align}
& {\mathbb E}_{P}\left[\sum_{n=1}^{N} I(x_n,y_n)\right] \\
& = \sum_{y_{1}^{N}, u_{1}^{N},x_{1}^{N}} P(y_{1}^{N}, u_{1}^{N},x_{1}^{N}) \sum_{n=1}^{N} I(x_n,y_n) \\
& = \sum_{n=1}^{N} P(x_n=0,y_n=1) \\
& = N (1-s)t
\end{align}
Therefore, the expected number of memory cells turning from state 1 to state 0 is $(1-s)tN$ asymptotically. The number of recorded bits is the cardinality of the high-entropy set $F_N$, which is $(1-s)H(t)N$ asymptotically.  It can be checked that this is exactly the capacity of 2-write binary WOM codes \cite{rivest82} \cite{shipilka12}. 
It is not difficult to verify that the capacity region of $t$-write binary WOM codes can also be achieved by using the proposed coding scheme $t$ times, for any integer $t$. Thus the proposed coding scheme is capacity-achieving.
\end{proof}

\begin{figure*}
\begin{align}
\label{pq_exp_bound}
&\frac{1}{N}\left|{\mathbb E}_{Q}\left[\sum_{n=1}^{N} I(x_n,y_n)\right] -
{\mathbb E}_{P}\left[\sum_{n=1}^{N} I(x_n,y_n)\right]\right|  = \frac{1}{N}\left|\sum_{y_{1}^{N}, u_{1}^{N},x_{1}^{N}} \left(Q(y_{1}^{N}, u_{1}^{N},x_{1}^{N})- P(y_{1}^{N}, u_{1}^{N},x_{1}^{N})
\right) I(x_n,y_n)\right| \notag \\
& \stackrel{(a)}{\leq} \frac{1}{N}\sum_{y_{1}^{N}, u_{1}^{N},x_{1}^{N}} \left|Q(y_{1}^{N}, u_{1}^{N},x_{1}^{N})- P(y_{1}^{N}, u_{1}^{N},x_{1}^{N})
\right| I(x_n,y_n) \notag \\
& \stackrel{(b)}{\leq} \frac{1}{N}\sum_{y_{1}^{N}, u_{1}^{N},x_{1}^{N}} \left|Q(y_{1}^{N}, u_{1}^{N},x_{1}^{N})- P(y_{1}^{N}, u_{1}^{N},x_{1}^{N})
\right|  \stackrel{(c)}{\leq} \frac{1}{N}\sum_{y_{1}^{N}, u_{1}^{N}} \left|Q(y_{1}^{N}, u_{1}^{N})- P(y_{1}^{N}, u_{1}^{N})
\right|   \stackrel{(d)}{\leq} 2\delta_N
\end{align}
\end{figure*}

\section{Conclusion}

\label{sec_conclusion}

The paper presents one WOM coding scheme based on source polarization. We prove that the proposed coding scheme is capacity-achieving.

%\nocite{*}
\bibliographystyle{IEEEtran}
\bibliography{the_bib}

\end{document}